\pdfoutput=1
\newif\ifFull
\Fulltrue
\newif\ifs
\sfalse

\ifFull
\documentclass[letterpaper,USenglish]{article}
\usepackage{amsmath,amsfonts}
\usepackage{fullpage}
\usepackage{amssymb,amsthm}

\newtheorem{lemma}{Lemma}
\newtheorem{theorem}{Theorem}

\else
\documentclass[sigconf, natbib]{acmart}
\renewcommand{\subsection}[1]{\paragraph{\bf #1.}}
\fi
\usepackage{hyperref,graphicx,subcaption}
\usepackage[noend]{algorithmic}
\ifFull
\usepackage{cite}
\fi

\renewcommand{\emph}[1]{\textbf{\textit{#1}}}%

\ifFull
\renewcommand{\cal}{\mathbf}
\else
\setcopyright{rightsretained}
\acmConference[ACM SIGSPATIAL 2017]{ACM SIGSPATIAL Int. Conf. on Advances in Geographic Information Systems}{Nov.~2017}{Redondo Beach, CA USA}
\copyrightyear{2017}
\acmISBN{978-1-4503-XXXXXX}\acmPrice{\$15.00}
\acmDOI{10.1145/XXXXX.XXXX}
\newcommand{\cal}{\mathbf}
\fi

\begin{document}
\clubpenalty=1000
\widowpenalty=1000
\tolerance=1000
\hyphenpenalty=500

\ifFull
\title{Answering Spatial Multiple-Set Intersection Queries Using 2-3 Cuckoo Hash-Filters}
\fi

\ifFull
\author{Michael T.~Goodrich \\
Dept. of Computer Science, Univ. of California, Irvine}
\fi

\ifFull
\date{}
\fi

\ifFull
\maketitle
\fi

\begin{abstract}
We show how to 
answer spatial multiple-set intersection queries in
$O(n(\log w)/w\, +\, kt)$ expected time,
where
$n$ is the total size of the $t\le w^c$ sets involved in the query,
$w$ is the number of bits in a memory word, $k$ is the output size,
and $c\ge 1$ is any fixed constant.
\ifFull
This improves the asymptotic performance over previous 
solutions and is based on an interesting data structure, known as
\emph{2-3 cuckoo hash-filters}.
Our results apply in the word-RAM model (or practical RAM model), 
which allows for constant-time bit-parallel operations, such
as bitwise AND, OR, NOT, and MSB (most-significant 1-bit), as exist in 
modern CPUs and GPUs.
Our solutions apply to any multiple-set intersection 
queries in spatial data sets
that can be reduced to one-dimensional range queries, such
as spatial join queries for
one-dimensional points or
sets of points stored along space-filling curves, 
which are used in GIS applications.
\fi
\end{abstract}

\section{Introduction}

\ifFull
\begin{figure}[hbtp]
\begin{center}
\includegraphics[width=2.8in, trim = 1.3in 0.5in 1.3in 0.4in, clip]{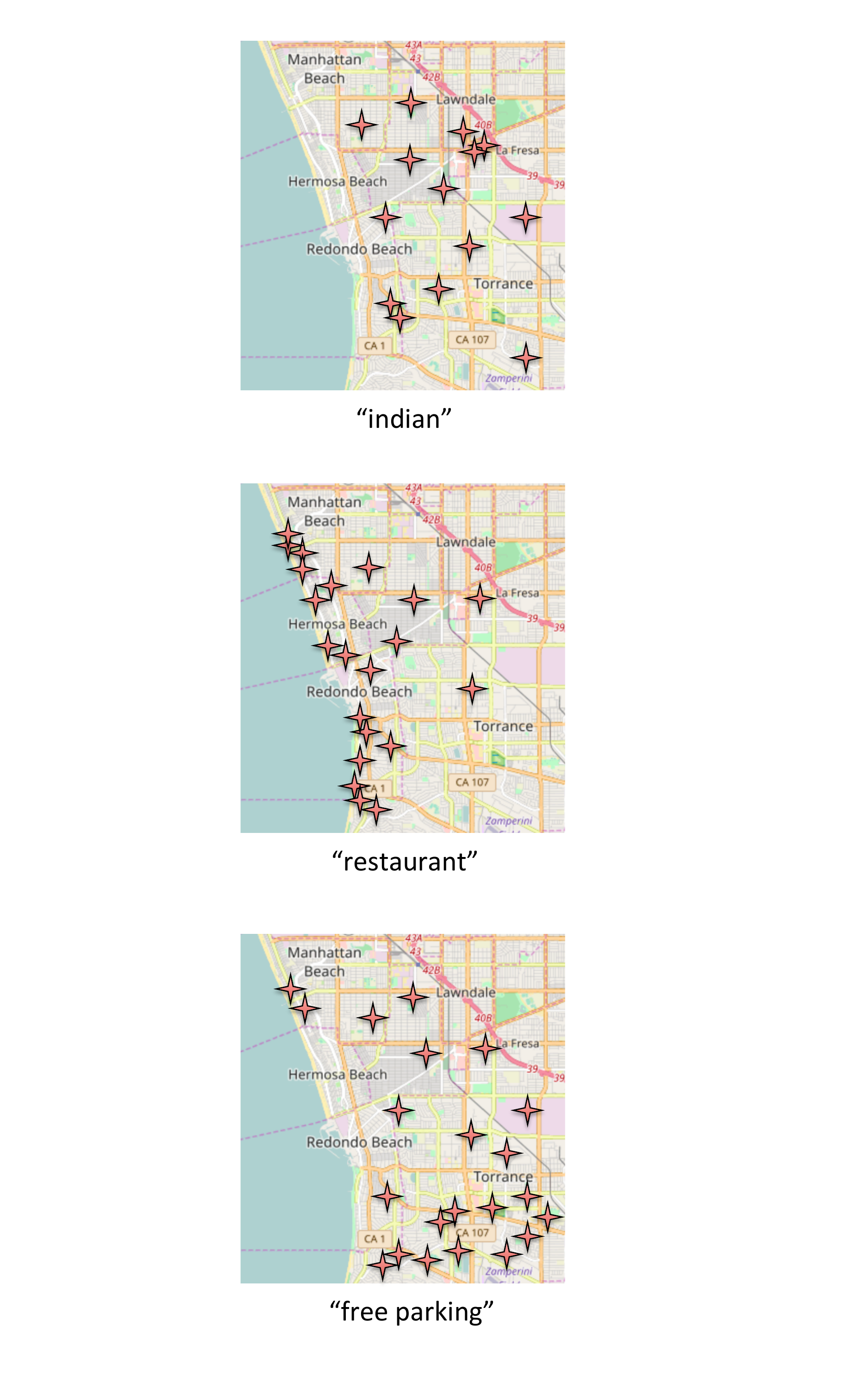}
\end{center}
	\vspace*{-10pt}
\caption{\label{fig:multi} An example spatial multiple-set intersection search. 
Each frame shows a set of possible responses for three different keyword
searches in a geographic region, which would then need to be intersected to
find sites matching all three keyword searches.
Background image Copyright~\copyright~OpenStreetMap contributors. 
Licensed as CC BY-SA.}
\end{figure}
\fi

\ifFull
Algorithms for
answering spatial multiple-set intersection queries 
have a number of different applications,
including in keyword-based location searching, and web searching.
E.g., 
see~\cite{Carver-1080,Chou2008293,Ji2011,Jiang-1080}.
For example, such algorithms can 
arise in the ``inner loop'' computation for answering 
conjunctive keyword queries in search engines.
Alternatively, we might have a data structure that 
stores sites matching certain keywords according to some spatial index.
Then, a user might issue a conjunctive keyword search 
and she might be interested 
in all of the sites close to a geographic region that match all of her search terms.
We are therefore interested in this paper in high-performance algorithms for
computing such spatial multiple-set intersection queries. 
See Figure~\ref{fig:multi}.

\subsection{Computational Model}
\fi

In this paper, we are interested in 
asymptotic improvements to spatial multiple-set intersection queries
by taking advantage of bit-level parallelism, in a
computational model known as the
\emph{practical RAM} model~\cite{Miltersen1996}
or \emph{word-RAM} model~\cite{Hagerup1998}.
By ``bit-level parallelism,'' we are referring to an ability to compute
bit-parallel operations on pairs 
of binary words of $w$ bits in constant time, e.g., using operations
 built into modern CPUs and GPUs.
There are actually different versions of the
{word-RAM} model 
\ifFull
(e.g., see~\cite{Hagerup1998}).
\fi
For example, we can consider a
\emph{restricted word-RAM} model~\cite{Hagerup1998}, 
where bit-parallel operations are limited
to addition, subtraction, and
the bit-parallel operations AND, OR, NOT, XOR, shift, and MSB 
(most-significant set bit).
\ifFull
We can also consider extensions to this model,
including a \emph{multiplication word-RAM} model, which
would also include constant-time multiplication,
and an \emph{AC$\,^0$ word-RAM} model, which would also include any
AC$\,^0$ operation, that is, any operation that can be computed with a 
constant-depth circuit with unbounded fan-in (which does not include
multiplication), e.g., see~\cite{Hagerup1998,Miltersen1996}.
\fi
In this paper, we provide results
for the restricted word-RAM model and also for a
\emph{permutation word-RAM} model (e.g., see~\cite{ALBERS199725}), which
\ifFull
can be viewed as somewhat weaker than the multiplication word-RAM,
in that the permutation word-RAM model 
\fi
can perform
a fixed permutation of $w/m$ subwords, each of size $m$, in constant time.
\ifFull
This is an operation supported by many modern CPUs and GPUs
(e.g., see~\cite{Hilewitz2008,ML03,SL03,SYL03,Yang99}).
As in the traditional RAM model, we analyze 
the running times of word-RAM algorithms
by counting the number of operations performed.
\fi
\ifFull
\subsection{Problem Formulation}

\fi
Formally, we assume we have a computational setting that 
consists of the following items:
\ifFull
\begin{itemize}
\item
\else

$\bullet$
\fi
A collection of data sets, $D_1,D_2,\ldots$, that
store \emph{items}, such that each item is associated
with a point along a one-dimensional curve, $\mathcal{C}$, which
is the same for all the data sets.
In the simplest case, the curve, $\mathcal{C}$, is just a straight line,
but we also allow $\mathcal{C}$ to be a space-filling curve, such
as a Hilbert curve or z-order curve.
\ifFull
E.g., see~\cite{sagan2012space} and
Figure~\ref{fig:curve}.
Such curves
are often used in spatial data indexing applications, 
e.g., see~\cite{ASANO19973,Lawder2000,Lawder:2001,Lee:2010,Liao01,Samet:1990:DAS:77589},
so as to reduce multi-dimensional approximate nearest-neighbor
and range queries to 1-dimensional range queries.
For example,
each $D_i$ could be a collection of items along a 1-dimensional
line or space-filling curve such that
each $D_i$ is associated with a certain category, such as
keyword matches for coffee shops or parking garages.
\fi
\ifFull
\item
\fi
A \emph{spatial multple-set intersection query} consists of 
an interval range, $\mathcal{R}\subseteq\mathcal{C}$, 
and a set of indices, $\mathcal{I}=\{i_1,i_2,\ldots,i_t\}$.
Each $i_j$ identifies a specific data set, $D_{i_j}$, e.g., based
on some keyword of interest.
The response to this query should be every item
that has a point in the interval range, $\mathcal{R}$, 
and belongs to the common intersection,
$
D_{i_1} \cap
D_{i_2} \cap
\cdots
\cap D_{i_t}$.
\ifFull
\end{itemize}
\fi

\ifFull
For instance, 
each $D_{i_i}$ could store
all the shops of a certain type in some district organized 
along a space-filling curve, $\mathcal{C}$,
so as to answer proximity queries that can
be expressed in terms of an interval range (or constant number of
interval ranges) along $\mathcal{C}$.
Then a spatial multiple-intersection query might ask 
for all the shops satisfying a set of
different keywords and be nearby some specified location, using an
interval-range query
that is amounts to a ``spatial join'' or ``distance join''
(e.g., see~\cite{Beckmann:1990,Hjaltason:1998,Huang:1997,WJ96}).
Alternatively,
in an even simpler scenario, 
$\mathcal{C}$ could just be the one-dimensional ``number line'' and
each $D_{i_j}$ could just be a set of 
items associated with points along this line.
Furthermore, the interval range, $\mathcal{R}$, could be the interval
$(-\infty,+\infty)$, in which case the problem becomes one of simply
returning the common intersection of a collection of sets, independent
of the curve, $\mathcal{C}$.
We are interested in answering such queries as quickly as possible,
given a reasonable amount of preprocessing.
\fi

\ifFull
\begin{figure}[hbt!]
\begin{center}
\includegraphics[width=2.3in]{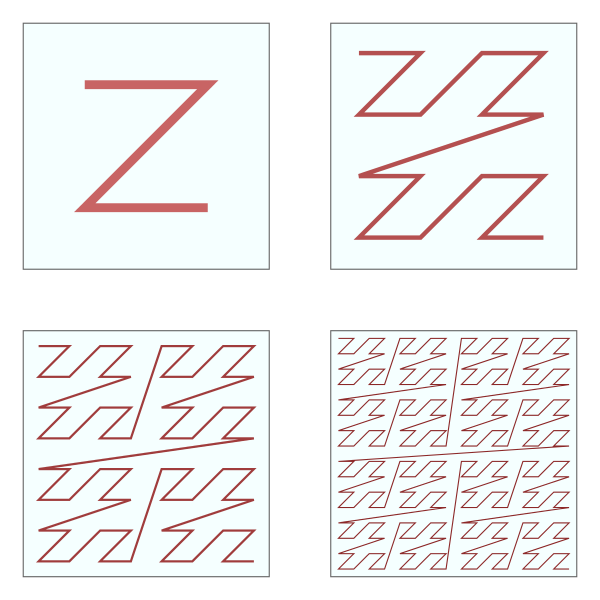}
\end{center}
\vspace*{-11pt}
\caption{An example z-order curve, at several resolutions,
which is related to the well-known quadtree spatial data structure.
Copyright~\copyright~David Eppstein, based on an image by Hesperian.
Licenced by CC A-S 3.0.
\label{fig:curve}}
\end{figure}
\fi

\subsection{Related Work}
In addition to work by
Miltersen~\cite{Miltersen1996} introducing the 
{practical RAM} model and work by Hagerup~\cite{Hagerup1998}
introducing the {word-RAM} model,
researchers have explored various algorithms
for versions of the practical RAM and word-RAM models,
\ifFull
e.g., see~\cite{Thorup:2003,Willard00,Fredman1993424,Andersson1999337}.
\fi
There is, for instance, considerable previous work on algorithms for answering 
set-intersection queries in the word-RAM model.
Ding and K{\"o}nig~\cite{Ding:2011}
show how to compute the common intersection of $t$
sets of total size $n$ in expected time
$O(n/\sqrt{w} + kt)$, 
where $w$ is the word size in bits and $k$ is the size of the output in words.
Bille {\it et al.}~\cite{Bille07} present 
a data structure that can compute the
intersection of $t$ sets of total size~$n$ in 
$O(n(\log^2 w)/w + kt)$ expected time in the permutation word-RAM model.
In addition,
Kopelowitz {\it et~al.}~\cite{Kopelowitz2015} introduce a 
data structure for computing set intersections
for two sets of roughly the same size, $n$, in
$O(n(\log^2 w)/w+\log w+k)$ expected time in this model.
Eppstein {\it et al.}~\cite{Eppstein2017} improve this bound to
$O(n(\log w)/w+k)$ expected time for the restricted word-RAM model,
using a data structure they call ``2-3 cuckoo hash-filters,'' 
 which consist of a combination
of 2-3 cuckoo hash tables and 2-3 cuckoo filters.
\ifFull
2-3 cuckoo hash tables use a generalization of the
\emph{power-of-two-choices}
paradigm~\cite{Mitzenmacher00thepower,Mitzenmacher01power,Lumetta07,Azar99}
to a \emph{two-out-of-three} paradigm.
That is, in a 2-3 cuckoo hash table, each item has three possible
pseudo-random places it can be stored and it is stored in two of 
them~\cite{Amossen11}.
\fi
Eppstein and Goodrich~\cite{EG17spaa} extend this contruction for 
pairs of sets of different sizes.
\ifFull
A {2-3 cuckoo filter}~\cite{Eppstein2017} parallels a 2-3 cuckoo
hash table storing the same items, except that we store only
a $\Theta(\log w)$-bit pseudo-random fingerprint, $f(x)$, for each item $x$
in the associated cuckoo filter.
\fi
Unfortunately, all of these previous uses of 2-3 cuckoo hash-filters
are limited to pairwise intersections of sets and they do not extend
to intersections of three or more sets or spatial queries.
Thus, the best previous algorithm for answering (standard) multiple-set intersection 
queries in the word-RAM model is due to Bille {\it et al.}, as mentioned above.
We are not familiar with any previous work in the word-RAM model 
for answering spatial multiple-set intersection queries.

\subsection{Our Results}
\ifFull
We present simple new data structures and algorithms for answering 
spatial multiple-set intersection queries in the word-RAM model based on 
using 2-3 cuckoo
hash-filters~\cite{Eppstein2017,EG17spaa}.
\begin{itemize}
\item
We show how to 
answer standard multiple-set intersection queries in
$O(n(\log w)/w\, +\, kt)$ expected time 
in the permutation word-RAM model,
where
$n$ is the total size of the $t$ sets involved in the query,
$w$ is the number of bits in a memory word, $k$ is the output size,
$t\le w^c$, and $c\ge 1$ is an arbitrary fixed constant.
This improves the asymptotic performance over the previous 
best multiple-set intersection method for the permutation word-RAM model, due to
Bille {\it et al.}~\cite{Bille07}, for $t$ being polynomial in $w$.
\item
We show how to 
answer standard multiple-set intersection queries in
$O(n(\log^2 w)/w\, +\, kt)$ expected time 
in the restricted word-RAM model,
where $n$ is the total size of the $t$ sets involved in the query,
$w$ is the number of bits in a memory word, $k$ is the output size,
$t\le w^c$, and $c\ge 1$ is an arbitrary fixed constant.
This matches the asymptotic performance of Bille {\it et al.}~\cite{Bille07},
for $t$ being polynomial in $w$,
but does so in a restricted word-RAM rather than the permutation word-RAM.
\item
\fi
We show how to 
answer spatial multiple-set intersection queries in
$O(n(\log w)/w\, +\, kt)$ 
expected time in the permutation word-RAM model,
or $O(n(\log^2 w)/w\, +\, kt)$ expected time 
in the restricted word-RAM model,
where $n$ is the total size of the $t\le w^c$ location-constrained 
subsets involved in the query, where $c\ge1$ is a fixed constant
and $k$ is the output size.
\ifFull
This is, to our knowledge, the first such result of its kind.
\end{itemize}
\fi

Our data structures and algorithms 
take advantage of 
a simple approach that exploits
bit-level 
parallelism by
packing information into memory subwords,
which allows us to represent small sets of size $O(w/\log w)$ with
2-3 cuckoo filters occupying 
$O(1)$ memory words, and intersect pairs of sets represented this way
in expected time
$O(1+k)$, where $k$ is the output size.
The main computational difficulty of using this approach is that the result of
such an operation is itself \emph{not} a 2-3 cuckoo filter.
Nevertheless,
unlike previous results that exploited 2-3 cuckoo filters for
intersection queries~\cite{Eppstein2017,EG17spaa},
we show in this paper how to restore the result of a pairwise
intersection query back to being a 2-3 cuckoo filter.
We then show that
such restoration operations 
allow us to achieve our results for 
spatial multiple-set
intersection queries.
For full details and proofs, please see the full version of this 
paper~\cite{GoodrichFull}.

\section{A Review of 2-3 Cuckoo Hash-Filters}

Let us
review the 2-3 cuckoo hash-filter data structure
of Eppstein {\it et al.}~\cite{Eppstein2017}.
Suppose we wish to represent a set, $S$, of $n$ items
taken from a universe such that each item can be stored
in a single memory word.
\ifFull
We assume throughout that $w\ge \log n$, as is standard.
\fi
We use the following components, $T$, $M$, $C$, and $F$.
\ifFull
(See Figure~\ref{fig:filters}.)
\fi
\begin{itemize}
\item
A hash table~$T$ of size $O(n)$,
using three pseudo-random hash functions $h_{1}$,
$h_{2}$, and $h_{3}$, which map items of $S$ to triples of distinct integers
in the range $[0,n-1]$.
Each item, $x$ in $S$, is stored, if possible, in two of the three possible
locations for $x$ based on these hash functions.
We refer to each $T$ as a 2-3 \emph{cuckoo hash table}.
\item
We also store a stash cache\ifFull~\cite{kirsch}\fi, 
$C$, of size $\lambda$, where $\lambda$ is bounded by a constant.
$C$ stores items for which it was not possible to store properly in 
two distinct locations in $T$\ifFull
 (e.g., due to collisions with other items)\fi.
We maintain each $C$ as an array of size $O(\lambda)$.
\item
A table, $F$, having $O(n)$ cells, that parallels $T$, so that $F[j]$
stores a non-zero 
fingerprint digest, $f(x)$, for an item, $x$, if and only if $T[j]$
stores a copy of $x$.
The digest $f(x)$ is a non-zero random hash function 
comprising $\delta$ bits, where $\delta=c\log w$ is a parameter
chosen to achieve a small false positive rate.
The table $F$ is called a 2-3 \emph{cuckoo filter}, and it is
stored in a packed format, so that we store $O(w/\log w)$ cells of
$F$ per memory word. 
In addition to the vector $F$, we store a bit-mask, $M$, that is
the same size as $F$ and has
all 1 bits in the corresponding cell of each occupied cell of $F$.
\end{itemize}
\ifFull
We assume we can read and write individual cells of $F$ and $M$
in $O(1)$ time.
These cells amount to subfields of words 
of $\delta=O(\log w)$ size, which can be read from or written to
using standard bit-level operations in the restricted word-RAM model.
Thus,
we assume that reading or writing 
any individual cell of $F$ takes $O(1)$ time in the 
restricted word-RAM model.
\fi

\ifFull
\begin{figure}[htb]
\begin{center}
\includegraphics[width=3.4in, trim = 1.3in 3.4in 1.6in 1.0in, clip]{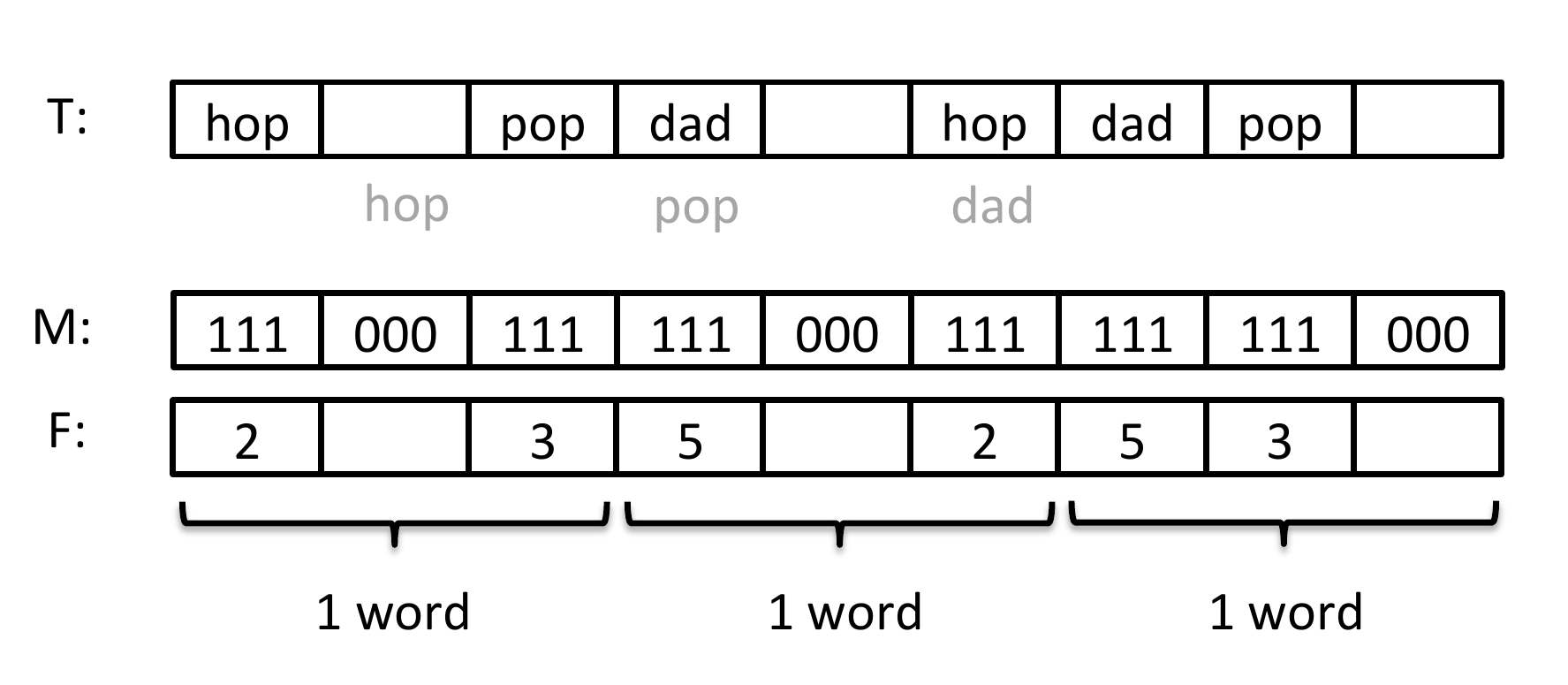}
\end{center}
	\vspace*{-18pt}
\caption{\label{fig:filters} A simplified example of 
a 2-3 cuckoo hash table, $T$, and filter, $F$, including the bit-mask,
$M$.
In this example, we are representing the set, 
$S=\{\mbox{\sffamily hop, pop, dad}\}$, in an instance
of the practical RAM model that can store three fingerprint values per word.
In this case, the filters are defined so
$f(\mbox{\sffamily hop})=2$,
$f(\mbox{\sffamily pop})=3$,
and
$f(\mbox{\sffamily dad})=5$.
Note: every item is stored in two out of three locations; we are showing
the third potential location for each item below its third location in grey.
In this figure, we are not showing the stash cache, $C$, 
which in this simplified example would be empty.
Also, note that if $n$ is $O(w/\log w)$, then,
although the size of the 2-3 cuckoo table, $T$, is $O(n)$,
the size of the corresponding 2-3 cuckoo filter, $F$, is $O(1)$.
Copyright~\copyright~Michael Goodrich.
}
\end{figure}
\fi

\ifFull
Since the method for constructing a 2-3 cuckoo filter, $F$, is the same as
that for a 2-3 hash table, $T$, that is parallel to it, without loss
of generality, let us describe how to construct $T$.
We assume we have $n$ items that need to be added to $T$ and that
$T$ has size that is at least $6(1+\epsilon)n$, for a constant $\epsilon>0$.
We also assume that we have a stash cache, $C$, of constant size,
$\lambda$.
Eppstein {\it et al.}~\cite{Eppstein2017} describe an algorithm
for constructing a 2-3 cuckoo hash-filter, which we review below.
Our algorithm for constructing a 2-3 cuckoo hash-filter, then,
involves performing $n$ insertions into our structures.
We begin with performing those insertions in $T$, so as to construct
a 2-3 cuckoo hash table.

Suppose we are inserting an item, $x$.
We first attempt to insert $x$ into the 2-3 cuckoo table, $T$,
performing a 2-3 cuckoo (two-out-of-three) insertion for $x$.
Following Amossen and Pagh~\cite{Amossen11}, let us consider
this as our inserting two instances of $x$ into $T$ via a
one-out-of-three cuckoo insertion.
That is,
each item $x$ has three possible locations,
$T[h_{1}(x)]$,
$T[h_{2}(x)]$,
and
$T[h_{3}(x)]$,
where it may be stored.
If one of these is empty, then we add $x$ to it, completing that instance
of inserting $x$. 
If this is the second insertion for $x$, then we are done inserting
$x$.
If none of these three cells is 
empty, we choose one of them at random and add $x$ to it, evicting its previous occupant~$y$.
We add $y$ to a temporary buffer queue, $Q$.

We then process the buffer, $Q$, while it is non-empty.
We take the next item, $y$, from $Q$.
We then read the cells,
$T[h_{1}(y)]$,
$T[h_{2}(y)]$,
and
$T[h_{3}(y)]$.
One of these cells may already store $y$; in this case, we consider the other two.
If one of these two is empty, we add $y$ to it, and we are done with~$y$.
If both of these cells are occupied, however, we choose one of them
at random, evict its previous occupant~$z$, and insert~$y$ into that
cell.
Then we add $z$ to $Q$.
We repeat this processing of $Q$ until we either succeed in emptying
all the items in $Q$ or we reach a \emph{stopping condition},
which is defined to be the condition that we have spent more than $L$
iterations processing $Q$ during this insertion, where $L$ is a
threshold parameter set in the analysis.
If this stopping condition occurs,
then we remove from $T$ each copy of an item in $Q$, 
and we add each such item to the stash, $C$.
This step takes $O(L+1)$ steps in the practical RAM model (possibly amortized,
if we are implementing $C$ as a standard growable table).
As we show, we can set $L$ to be logarithmic in $n$
and set a constant threshold $s>0$ such that, 
with high probability throughout the process, $\lambda\le s$.
Thus, with high probability, our 2-3 cuckoo hash filters has stashes of constant size.
This property provides a ``safety net'' for our data structures:
with low probability, they could always
fall back to a standard hash table or sorted linked list to store their
sets, but with high probability, our structures will be much faster than this.
If the construction of $T$ succeeds, then we create a parallel (i.e.,
mirrored) copy
of $T$ as a hash filter, $F$, by replacing each item with its
fingerprint of size $\delta=O(\log w)$ and compressing every
block of size $O(w/\log w)$ words in $T$ into a single work for $F$.
If the construction of $T$ fails, however, then we instead sort the
items of the original set, $S$,
of $n$ items according to some ranking function (e.g., page rank)
and just use this sorted copy of $S$ to represent $S$.
We say that the 2-3 cuckoo hash-table for $S$ ``failed''
if we did not successfully construct $T$ using the above algorithm.
\fi

\ifFull
For completeness, we include an analysis of 2-3 cuckoo hash-tables in an appendix,
showing that we can build such a data structure of size, $n$, with
a stash of size $s$, in expected $O(n)$ time and that the probability 
that this construction fails is at most $\tilde O(n^{-s})$, where
the $\tilde O(\cdot)$ notation ignores polylogarithmic factors.
\fi

\section{Data Structures}
\ifFull
In this section, we describe our 
data structure framework.
Recall that we assume that we are given
a collection of data sets, $D_1,D_2,\ldots$, that
store items, such that each item is associated
with a point along a one-dimensional curve, $\mathcal{C}$, which
is the same for all the data sets.
In addition, we assume that any
spatial multple-set intersection query consists of 
an interval range, $\mathcal{R}\subseteq\mathcal{C}$, 
and a set of indices, $\mathcal{I}=\{i_1,i_2,\ldots,i_t\}$, and
the response to this query should be every item
that has a point in the range, $\mathcal{R}$, and belongs to the common intersection,
$
D_{i_1} \cap
D_{i_2} \cap
\cdots
\cap D_{i_t}$.
\fi

As a starting point for our data structure construction, 
we subdivide each data set, $D_i$,
into a sequence of interval regions along the curve, $\mathcal{C}$,
such that each region stores $\Theta(w/\log w)$ points.
Our data structure construction, then, is as follows.
	For each interval region, $R$, in one of our 
        structures, $D_i$, use the algorithm given above to
		construct a 2-3 cuckoo hash-filter for 
        the $\Theta(w/\log w)$ points in $R$ with a stash of
		constant size, $\lambda$. If the construction for $R$
		fails, then simply fallback to representing this subset 
		as a sorted listing of its items (according to some
       canonical ordering).
The resulting set of interval regions and the 2-3 cuckoo hash-filter 
or sorted list for each
region comprises our representation for this structure.
In spite of this failure possibility (which amounts to our using
a standard list-based subset representation as a fallback), our construction
has the following property.

\begin{lemma}
	\label{lem:fragile-size}
	Let $S$ and $T$ be sets of consecutive interval 
regions from two data 
sets, $D_i$ and $D_j$,
and let $n$ denote the number of points in all the regions of 
$S$ and $T$.
Let $\alpha$ denote
the number of interval region pairs, $(R_i,R_j)$, such that
$R_i$ and $R_j$ overlap and our 2-3 cuckoo hash-filter construction failed
for either $R_i$ or $R_j$. Then 
\ifFull
\[
{\mathbf E}[\alpha] \le \frac{2n\log w}{w^{s+1}},
\]
\fi
where $s\ge 2$ is a chosen constant.
\end{lemma}
\ifFull
\begin{proof}
Since each interval region in our data structures contains
$\Theta(w/\log w)$ points,
\ifFull
we can use Theorem~\ref{thm:failure},
from the analysis given in the appendix,
\else
we can use the previous analysis~\cite{Eppstein2017}
\fi
	to determine a stash constant, $\lambda$, for
	each 2-3 cuckoo hash-filter so that its failure
	probability is at most 
	\[
		\lceil w/\log w\rceil^{-3s} \le w^{-s},
	\]
	for $w\ge 8$, which follows from the fact that $w\ge \log n$.
Let $\chi_{i,j}$ be an indicator variable that is $1$ if and only if
an interval region $R_i$ from $S$ overlaps an interval region $R_j$ from $T$, and let
$X(R)$ be a random variable that is $1$ if and only if our 2-3 construction
for the region $R$ failed.
Then, by the linearity of expectation,
\begin{eqnarray*}
{\mathbf E}[\alpha] &\le& \sum_{i,j} \chi_{i,j}(X(R_i)+X(R_j)) \\
&\le& \frac{n\log w}{w} \cdot \left(\frac{1}{w^s} + \frac{1}{w^s}\right) \\
&=& \frac{2n\log w}{w^{s+1}}.
\end{eqnarray*}
\end{proof}

	Thus, the 
	total expected number of items in all pairs of
regions such that one of the two regions has a failed 2-3 cuckoo hash-filter
	is at most proportional to
	\[
		\frac{n\log w}{w} \cdot \frac{w}{\log w} \cdot \frac{1}{w^s}.
	\]
\fi
	That is, the total expected 
	number of points summed across all pairs of overlapping
regions where one of the regions has a failed 2-3 cuckoo hash-filter is
	$O(n/w^s)$.
\ifFull
In other words, since we are choosing $s\ge2$, we can use a fallback 
set-intersection method based on merging sets sorted according to some canonical 
order and the expected cost of all such intersections will be $O(n/w^s)$.
\fi

\ifFull
There is still one more component to our construction, which is a set of
structures for each region that we call \emph{cuckoo-restore} structures,
which allow us to restore an intersection representation
to be a 2-3 cuckoo hash-table for any interval. 
Since the motivation 
for these structures depends on the method for doing a pairwise intersection
computation, let us postpone our discussion of the cuckoo-restore structures
until after we have given our pairwise intersection algorithm.
\fi

\section{Intersection Algorithms}
\label{sec:intersect}
In this section,
we describe our algorithm for performing 
spatial multiple-set intersection queries.

\subsection{Intersecting Two 2-3 Cuckoo Hash-Filters}
Let us begin by reviewing the method of 
Eppstein {\it et al.}~\cite{Eppstein2017}
for intersecting a pair of 2-3 cuckoo hash-filters, which 
in our case will always be represented using $O(1)$ memory words
for the filter component, since interval regions in our
structures hold $\Theta(w/\log w)$ items.

Suppose then that we have two subsets, $S_i$, and $S_j$, of size $O(w/\log w)$
each for which we wish to compute a representation of the intersection
$S_i\cap S_j$.
Suppose further that $S_i$ and $S_j$ are each represented with 2-3
cuckoo hash-filters of the same size and using the same three hash
functions and fingerprint function.
We begin our set-intersection algorithm
for this pair of 2-3 cuckoo hash-filters
by computing a vector of $O(1)$ words
that identifies the matching non-empty cells in
$F_i$ and $F_j$.
For example, we could compute the vector defined by
the following bit-wise vector expression:
\begin{equation}
A = (M_i ~ {\rm AND ~ NOT} ~ (F_i ~ {\rm XOR} ~ F_j)).
\label{eq:filter}
\end{equation}
We view $A$ as being a parallel vector to $F_i$ and $F_j$.
Note that a cell, $A[r]$, consists of $\delta$ bits
and this cell is all 1s if and only if $F_i[r]$ stores a
fingerprint digest for some item and $F_i[r]=F_j[r]$,
since fingerprint digests are non-zero.
Thus, with standard operations in the word-RAM we can compute a
compact representation of $O(1)$ words that stores each subword
for each fingerprint that matches in $F_i$ and $F_j$, with the matching
locations identified by all 1's in a mask, $M_i'$.
The crucial insight is that since each item is stored in two-out-of-three
locations in a 2-3 cuckoo hash-filter, if the same item is stored in two
different cuckoo hash-filters, then it will be stored in 
one of its three locations in both hash-filters. It is this common location
that then stores the filter for this item after we perform the bit-parallel
operations to compute this intersection.

Thus, if we are interested in just this pairwise intersection,
we can create a list, $L$,
of members of the common intersection of $S_i$ and $S_j$,
by visiting each word of $A$ and storing to $L$ the item in
$T_i[r]$ corresponding to each cell, $A[r]$, that is all 1s, but
doing so only after confirming that $T_i[r]=T_j[r]$.
In addition, since we are assuming that stashes are of constant size,
we can do a lookup in the other hash table for each item in a stash
for $S_i$ or $S_j$, which takes an additional time that is $O(1)$.
Therefore,
the listing of the members in $S_i\cap S_j$
can be done in time $O(1+k+p)$, where $k$ is the 
number of items in the intersection and $p$ is the number of 
false positives (i.e., places where fingerprints match but the item
is not actually in the common intersection), 
by using standard and bit-parallel operations in the word-RAM
model (e.g., see also~\cite{Eppstein:2016}).

Note that the additional step of weeding out false positives can be done
at the end, which involves
a constant-time operation
per item in the list,
that itself involves 
a lookup in the two cuckoo hash tables, to remove items that
map to the same locations and have the same fingerprint digests 
but are nevertheless different items.
\ifFull
That is, we remove from this list any items, $x$ and $y$,
that happen to map to the same cell, $r$, in their respective 2-3 cuckoo
hash tables and they also have the same fingerprint digest, that is,
$f(x)=f(y)$.
\fi
Note that by requiring $\delta=c\log w$, we can guarantee that the
probability of such false positives is at most $1/w^c$.
\ifFull
Figure~\ref{fig:sets} illustrates a simplified 
example of this pairwise intersection algorithm.

\begin{figure}[htb!]
	\vspace*{-2pt}
\begin{center}
\includegraphics[width=3.2in, trim = 1.2in 1.5in 1.7in 1.5in, clip]{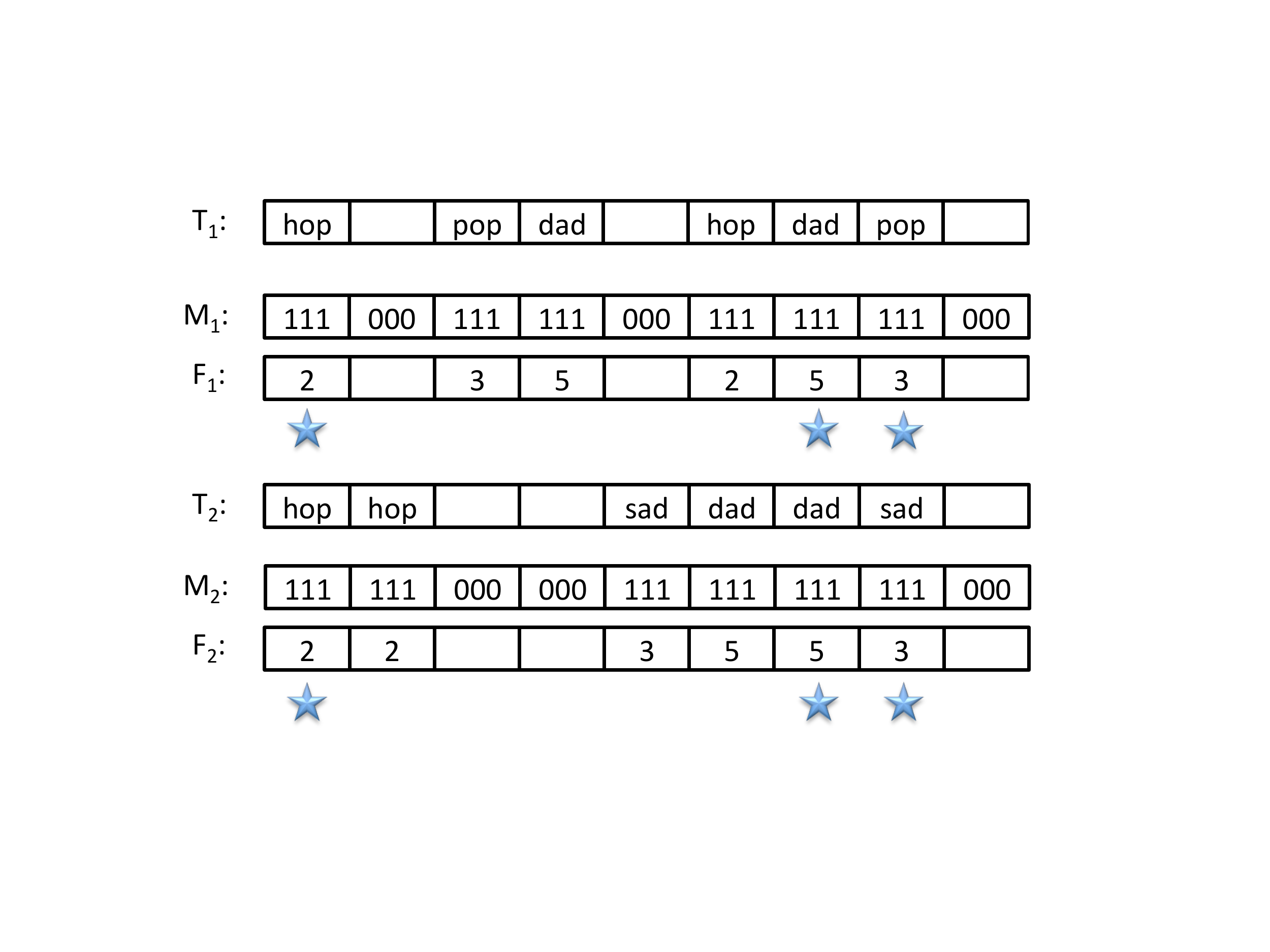}
\end{center}
	\vspace*{-20pt}
\caption{\label{fig:sets} An example of 
parallel 2-3 cuckoo hash tables and filters for two sets.
In this example, we are representing the set, 
$S_1=\{\mbox{\sffamily hop, pop, dad}\}$, 
and the set,
$S_2=\{\mbox{\sffamily hop, dad, sad}\}$.
The intersection algorithm does bit-parallel operations to find the matching
non-zero cells in $F_1$ and $F_2$, and then checks the matches found using
$T_1$ and $T_2$, to identify the set of intersecting items, 
$\{\mbox{\sffamily hop, dad}\}$.
In this case the filters are defined so
$f(\mbox{\sffamily hop})=2$,
$f(\mbox{\sffamily pop})=3$,
$f(\mbox{\sffamily dad})=5$,
and
$f(\mbox{\sffamily sad})=3$.
We show the matching filter locations with stars, including
the false positive match for
$\mbox{\sffamily dog}$
and
$\mbox{\sffamily fox}$,
which would be culled when we check the potential
matches against $T_1$ and $T_2$.
In this figure, we are not showing any stash caches.
Copyright~\copyright~Michael Goodrich.
}
\end{figure}
\fi

The running time of our entire algorithm for computing the intersection
of $S_i$ and $S_j$, therefore, is
$
O(1 + k + p)$,
where $k$ is the size of the output and $p$ is the number of false positives.
In addition, by choosing $\delta=c\log w$,
we can bound the probability that two different items have the same
fingerprint value as being at most $1/w^c$.
Thus, ${\bf E}[p]\le n/w$, since each item is stored in at most two places
in a 2-3 cuckoo hash table.
\ifFull
Therefore, we have the following.

\begin{theorem}
\label{thm:intersect}
Let $S_i$ and $S_j$ be two subsets of size $\Theta(w/\log w)$
each, represented by
using 2-3 cuckoo hash-filters with constant-size stashes. 
Then we can compute the intersection $S_i\cap S_j$ in 
$O(1 + k)$ expected time,
in the restricted word-RAM,
where $k$ is the size of the intersection.
\end{theorem}
\fi

\subsection{Answering Pairwise Spatial Set-Intersection Queries}
Let us next describe our algorithm for answering
a spatial two-set intersection query, asking for the intersection of two sets,
$S_1$ and $S_2$, of possibly different sizes.
Let us assume that $S_1$ and $S_2$ are each represented using
the structures as described above, that is, $S_1$ and  $S_2$ 
are subdivided into interval regions, such that,
for each interval region, we have done our construction of a 2-3 cuckoo
hash-filter (or, if that failed, then we have a sorted list of the items
in that interval).
In addition, we assume that we also 
are representing each entire set, $S_i$, for $i=1,2$,
using a standard hash table, $H_i$, such as a cuckoo hash table.
This hash table, $H_i$, will allow us to cull false positives as a post-processing
step.

Given a query interval range, $\mathcal{R}\subseteq\mathcal{C}$,
we first cull from $S_1$ and $S_2$ all the intervals that do not intersect
$\mathcal{R}$.
To allow for easier analysis of our method
for computing the intersection of $S_1$ and $S_2$,
let $n_1$ denote the number of items remaining in $S_1$ and let $n_2$ 
denote the number of items remaining in $S_2$, and let $n=n_1+n_2$.
Note that, since
the number if items in each remaining interval region in either $S_1$ 
or $S_2$ is $O(w/\log w)$,
the number of intervals in each $S_i$ is $O(n_i(\log w)/w)$, for $i=1,2$.
In addition, we assume that
we are representing these interval regions so 
that have an easy way of identifying when two regions overlap.

The goal of our algorithm 
is to compute a cuckoo hash-filter representation that contains
all the items in $S_1\cap S_2$, plus possibly some false positives
that our algorithm identifies as high-probability items belonging
to this common intersection.
\ifFull
After we have performed the core part of algorithm, then, we 
can simply do a lookup in each $H_i$ 
to confirm which items actually
belong to the common intersection (and should be produced in an output response)
and which items should be 
ignored because they are false positive members of the
common intersection.
\fi

Because of the way our algorithm works, it produces a cuckoo-filter
representation for the common intersection, where $S_1\cap S_2$ is represented
in terms of the intervals in $S_2$ (or, alternatively, we can swap the two
sets and our output will be in terms of the intervals of $S_1$).
Namely, for each such
interval, $I$, we will have a cuckoo
filter, $F_I$, and a backing 2-3 cuckoo table, $T_I$, where, for each 
non-zero fingerprint, $F_I[j]$, 
there is a corresponding item, $x=T_I[j]$, with
$x$ being a confirmed item in $S_2$.
The cuckoo filter, $F_I$, may not be a 2-3 cuckoo filter, however, because
each item might be stored in just one location in $F_i$, not two.
(We explain later how we can repair this situation to quickly build a 2-3 cuckoo
filter representing the pairwise intersection, albeit possibly with a small number
of false positives that can be removed in a post-processing step.)
Nevertheless, after our core algorithm completes,
our representation allows us to examine each non-zero cuckoo filter
fingerprint in a filter $F_I$, lookup its corresponding item, $x$, 
in a backing 2-3 cuckoo table, $T_I$, and
then perform a search for $x$ in the hash table for the other set (e.g.,
$S_1$) to verify that it belongs to the common intersection or is a false
positive.
That is, after one additional lookup 
for each such candidate intersection item, $x$, we can confirm 
or discard $x$ depending on whether it is
or isn't in the common intersection and with one more comparison.

There is also a possibility that our construction of a 2-3 cuckoo hash-filter
fails for some interval, in which case we fallback to a standard intersection
algorithm, such as merging two sorted lists, or looking up each item in 
one set 
in a hash table for the other. Since such failures occur with probability
at most $1/w^s$, for some constant $s$, the time spent on such fallback 
computations is dominated by the time for our other steps; hence, let us
ignore the time spent on such fallback computations.

Our algorithm
for constructing the representation of $S_1\cap S_2$,
using 2-3 cuckoo hash-filters,
and then optionally culling out false positives,
is as follows.
\ifFull
(See Figure~\ref{fig:intersect}.)

\begin{figure}[htb]
\vspace*{-10pt}
\begin{center}
\includegraphics[width=3.5in, trim = 0.5in 1.0in 1.5in 0.4in, clip]{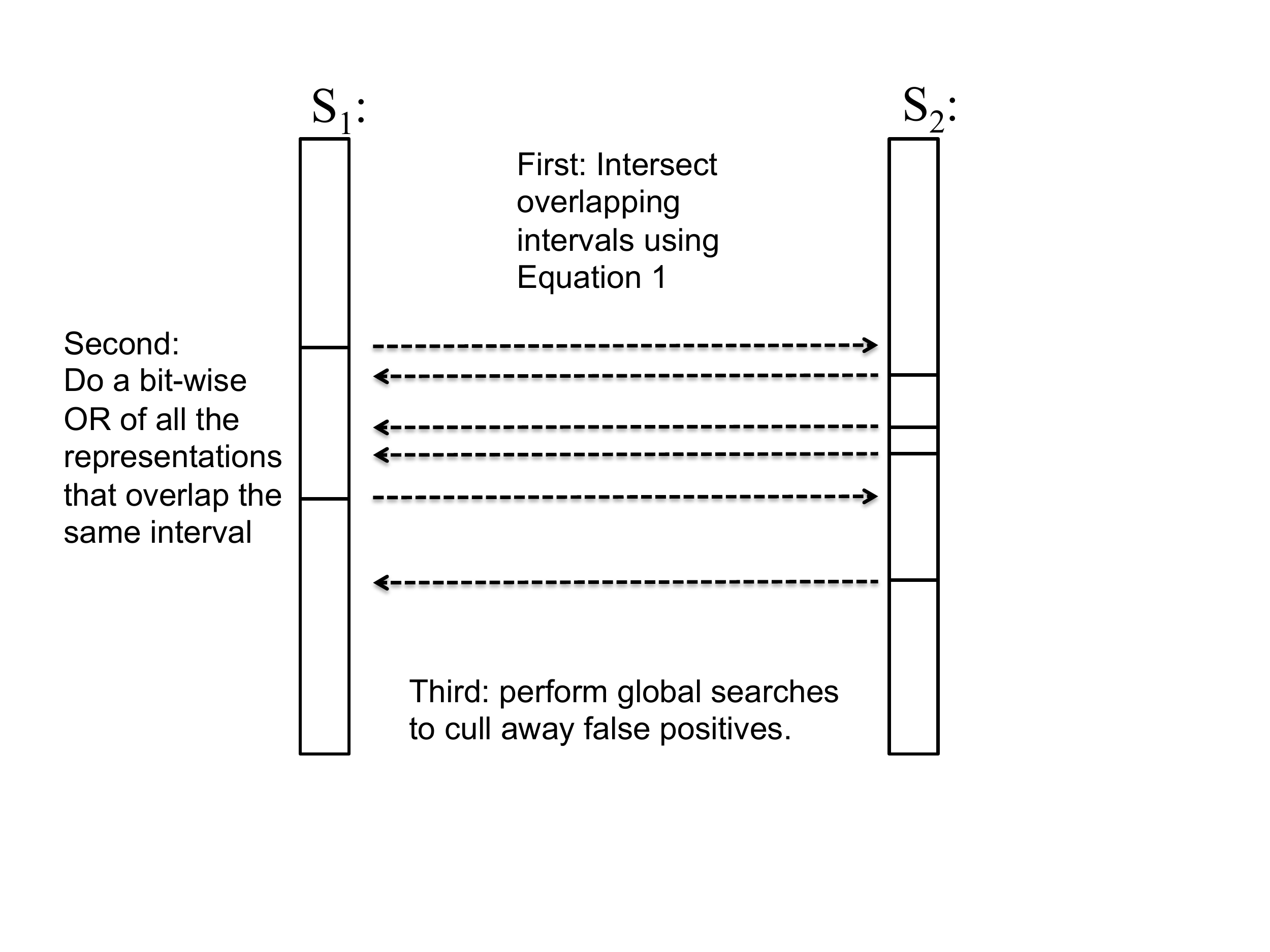}
\end{center}
\vspace*{-18pt}
\caption{\label{fig:intersect} An illustration of
	how we answer set-intersection queries for 
	a pair of sets subdivided into one-dimensional interval regions.
Copyright~\copyright~Michael Goodrich.}
\end{figure}
\fi

\begin{enumerate}
	\item
		Merge the interval regions
		of $S_1$ and $S_2$, to identify each pair of
   overlapping interval regions. 
		This step can be done in 
		$O(n(\log w)/w)$ time,
		by Lemma~\ref{lem:fragile-size}.
	\item
		\label{step:and1}
		For each overlapping interval, $I_{1,j}$ and $I_{2,k}$, where
		$I_{1,j}$ is from $S_1$ and $I_{2,k}$ is from $S_2$, 
		intersect these two subsets using the bit-parallel intersection 
		algorithm for 2-3 cuckoo filters 
		derived using Equation~\ref{eq:filter} above
		(but skipping the lookups in
		the corresponding 2-3 cuckoo table)
		Also perform lookups for
		any items in stashes (which are confirmed as
		intersections; hence, we are done with our computations
		for them).
		Let $F_{j,k}$ denote the resulting (now partial) 2-3 cuckoo
		filter.
		This step can be implemented
		in $O(n(\log w)/w)$ time, 
		since the total number of overlapping pairs
		of intervals is $O(n(\log w)/w)$.
	\item
		\label{step:and2}
		For each interval, $I_{1,j}$, collect all the 
		partial 2-3 cuckoo
		hash-filters, $F_{j,k}$, computed in the previous step
		for $I_{1,j}$.
		Compute the bit-wise OR of these filters. Let $F_j$ denote
		the resulting partial 2-3 cuckoo filter for $I_{1,j}$,
		and let ${\cal F}$ denote the collection of all such
		filters (note that there is potentially a non-empty
		partial cuckoo filter, $F_j$, for each interval
		in $S_1$).
		Also note that there are no collisions in the bit-wise
		OR of all these cuckoo filters, since each is computed as an
		intersection of a disjoint set of other items with the items
		in this interval.
		This step can be implemented
		in $O(n(\log w)/w)$ time, 
		since the total number of overlapping pairs
		of intervals is $O(n(\log w)/w)$.
	\item
		\label{step:and3}
		For each interval $I_{1,j}$, and
		each item, $x$, in the 2-3 cuckoo hash table 
		for $I_{1,j}$ that has a corresponding fingerprint belonging
		to $F_j$ in $\cal F$, 
		do a lookup in each of corresponding backing hash tables
(we assume we have global lookup tables for $S_1$ and $S_2$)
		to determine if $x$ is indeed
		a common item in the sets $S_1$ and $S_2$.
		Let $Z$ denote the set of all such items so determined
		to belong to this common intersection.
		This step can be implemented in 
		$O(n(\log w)/w+k+p)$ time, where $k$ is the size
		of the output and $p$ is the number of false
		positives, that is, items that have a non-zero
		fingerprint in some $F_j$ but nevertheless are not in the
		common intersection, $S_1\cap S_2$.
\item
	Output the members of the set, $Z$, as the answer.
\end{enumerate}

\ifFull
Let us consider the correctness of this algorithm.
First, note that
since each item included in the final output, $Z$, 
is confirmed to belong to the
common intersection, $S_1\cap S_2$, there are no false
positive items reported, if we do the optional culling step
(i.e., there are no reported items that
are not in the common intersection).
Thus, we need only show that each item in the common intersection is
added to $Z$.
The only possible way we could miss a member, $x$, 
of the common intersection
is if $x$ belongs to intervals in both $S_1$ and $S_2$, and we failed
to add $x$ to $Z$.
In this case, $x$ has a matching fingerprint in at least one common
location in the ultra-compact 2-3 cuckoo filter for two overlapping
intervals, which we determine in Steps~\ref{step:and1} 
and~\ref{step:and2}, and we confirm and add
to $Z$ any such overlapping items belonging to $S_1\cap S_2$
in Step~\ref{step:and3}.
Thus, the above algorithm is correct.
\fi

Let us analyze the running time of this algorithm.
We have already accounted for the time bounds for each step above.
Furthermore, note that 
the parameter $p$ is $O(n/w)$, since the probability of two
fingerprints of size $2\log w$ collide is at most $1/w$.
Therefore, the total expected time for our algorithm
is $O(n(\log w)/w + k)$. 
\ifFull
This gives us the following.

\begin{theorem}
	Given two sets, $S_1$ and $S_2$, represented with interval  regions
intersecting a query range $\mathcal{R}\subseteq \mathcal{C}$,
	with each region represented as described above, one can compute a listing
	of the items in $S_1\cap S_2$ in 
	$O(n(\log w)/w + k)$ expected time in the restricted word-RAM model,
	where $n$ is the size of $S_1$ and $S_2$ and $k$ is the size of the output.
\end{theorem}
\fi


\subsection{Answering Spatial Multiple-Set Intersection Queries}

The main bottleneck for extending our method from the previous section
to spatial multiple-set intersection queries is that the
partial result of performing the bit-parallel intersection of a pair of 2-3
cuckoo filters is itself not a 2-3 cuckoo filter, because after the
intersection computation is performed some fingerprints might be stored in just one
location, not two.
We can assume, however, that every subword in a fingerprint vector, $F$,
is either all 0's or it holds a complete fingerprint of $O(\log w)$ bits.

We can to augment our 2-3 cuckoo filter
representations so that we can restore them to be proper 2-3 cuckoo filters
even after a pairwise intersection computation, 
as follows:
\begin{enumerate}
\item
For each fingerprint vector, $F$, create and store \emph{cuckoo-restore}
information, which encodes a permutation, $\pi_F$, which routes every subword
fingerprint, $f$,
in $F$ to the location of its other location, which we call $f$'s \emph{twin}.  
That is, if the fingerprint $f$
for some item $x$ is stored at subword locations $i$ and $j$ in $F$, then
$\pi_F$ moves the copy of $f$ at $i$ to position $j$ and the copy of $f$ at $j$
to position $i$.
This step can be done in $O(1)$ time in the permutation word-RAM model.
\item
After an intersection operation occurs, so that each item in the intersection 
of a pair of 2-3 cuckoo filters
(including some possible false positives) may be is stored in a fingerprint
vector, $F$, in just one location instead of two,
apply $\pi_F$ to create a copy, $F'$ of $F$ such that each fingerprint 
subword in $F$ is routed to its twin location.
\item
Restore $F$ to be a 2-3 cuckoo filter by computing the bit-wise OR of $F$ and
$F'$. Note that there can be no collisions occurring as a result of
this bit-wise OR, because every fingerprint (even the false positives) 
is the same as its twin and all empty locations are all 0's. Thus,
this bit-wise OR will copy each surviving fingerprint to its 
twin location and then OR this fingerprint with itself (causing no change) or
with a subword of all 0's (restoring the fingerprint to its original two
locations).
\end{enumerate}

Thus, to answer a spatial multiple-set intersection
query, for sets, $S_1,S_2,\ldots,S_t$, we perform the above pairwise
intersection operations iteratively, first with $S_1$ and $S_2$, and then
with the result of this intersection with $S_3$, and so on, postponing until
the very end our culling of false positives by doing a lookup in global
hash tables for $S_1,S_2,\ldots,S_t$, which we assume we have available,
to check for each item $x$ in the final 2-3 cuckoo hash-table whether $x$
is indeed a member of every set.
Each intersection step, for an iteration $i$, 
takes $O(1)$ time per cuckoo-filter, in
the permutation word-RAM model, for which
we can charge this cost in iteration $i$ 
to the size of the subset in a region of the set $S_i$.

We can force 
the expected total number
of false positives to be at most $n/w^{c+1}$, for a fixed constant $c\ge 1$,
by defining the fingerprints to have size at least $(c+1)\log w$.
Thus,
testing each surving candidate at the end to see if it really belongs to the
common intersection, by looking up each such element $x$ in the $t$
hash tables for each set, takes time $O(kt)$ plus a 
term that is dominated by $O(n/w)$.

Note that
in order to implement our algorithm in the restricted word-RAM model,
the only part of this computation left as of yet unspecified is how to create a
representation of the permutation, $\pi_F$, and perform the routing of subwords
defined by $\pi_F$.
For this part,
we follow the approach of Yang {\it et al.}~\cite{Yang99}, who define
subword permutation micro-code instructions and show how to implement
them using a double butterfly network, which is also known as
a Benes network.
\ifFull
See Figure~\ref{fig:benes}.

\begin{figure}[htb!]
\vspace*{-10pt}
\begin{center}
\includegraphics[width=3in]{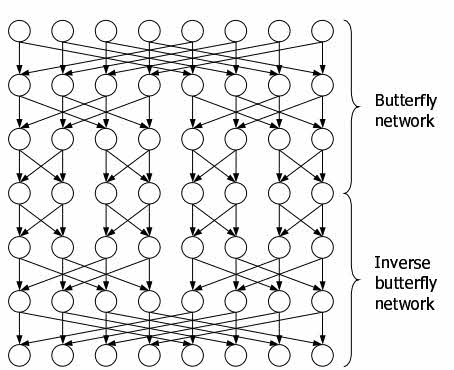}
\end{center}
\vspace*{-10pt}
\caption{An example Benes network, which is
constructed from two back-to-back butterfly networks.
Copyright~\copyright~David BS, licensed under cc by-sa 3.0.
\label{fig:benes}}
\end{figure}
\fi

It is known that a Benes network can route any permutation and that it is fairly
straightforward to set the switches in a Benes network for this purpose for any
given permutation. So let $N$ be such a 
network, which in the case of routing $O(w/\log w)$ subwords of size $O(\log w)$ 
will have depth $O(\log w)$.
Depending on how the switches are set,
we note that each stage of $N$ involves keeping some subword in place and moving
others by shifting them all the same distance. Thus, each stage of a Benes
network can be implemented in $O(1)$ steps in the word-RAM model by simple
applications of AND, OR, and shift operations (plus either 
the use of mask vectors to encode the switch settings or using other 
built-in word-RAM operations based on an encoding of $\pi_F$.
Thus, we can store an encoding of a Benes network implementing $\pi_F$
for each fingerprint vector, $F$, as our cuckoo-restore information, and this
will allow us to restore any cuckoo filter to be a 2-3 cuckoo filter in
$O(\log w)$ time in the word-RAM model.
This gives us the following result.

\begin{theorem}
Suppose $t\le w^c$ sets are from interval regions identified 
through a spatial multiple-set
intersection query for a range $\mathcal{R}\subseteq\mathcal{C}$,
in structures constructed as described above.
We can compute the result of such a query in
$O(n(\log w)/w + kt)$ expected time in the permutation word-RAM model,
or
$O(n(\log^2 w)/w + kt)$ expected time in the restricted word-RAM model,
where $n$ is the total size of all the sets involved,
$k$ is the size of the output,
and $c\ge 1$ is a constant.
\end{theorem}

\ifFull
\section{Discussion}
We should mention
a slight technicality with our results, as described above,
in that our methods report all the items
having points belonging to intervals intersecting a query range
$\mathcal{R}\subseteq\mathcal{C}$ 
and are in the common intersection.
Thus, there is a possibility for items close to one of the boundaries
of $\mathcal{R}$ but nevertheless outside of $\mathcal{R}$ 
to be included. This, of course, poses no computational difficulty, however,
as we could add a post-processing step that only outputs items with points
that are confirmed to be in $\mathcal{R}$ rather than points belonging to
intervals that intersect $\mathcal{R}$, if that is desired.
\fi

\paragraph{\bf Acknowledgments.}
This article reports on work supported by the 
DARPA under agreement no.~AFRL FA8750-15-2-0092.
The views expressed are those of the authors and do not reflect the
official policy or position of the Department of Defense
or the U.S.~Government.
This work was also supported in part from NSF grants
1228639, 1526631, 
1217322, 1618301, and 1616248.
We thank David Eppstein for several helpful discussions.

\ifFull
\bibliography{refs}
\fi

\ifFull
\clearpage
\begin{appendix}
\section{An Analysis of 2-3 Cuckoo Hash-Filters}
Following the framework 
of Eppstein {\it et al.}~\cite{Eppstein2017},
in this appendix
we analyze the performance of our construction algorithm using a 
hypergraph model
of 2-3 cuckoo hashing, which itself is a generalization of 
an analysis framework for (standard) cuckoo hashing,
as proposed by Pagh and Rodler~\cite{Pagh:2014}. 
In the standard cuckoo hashing case, one inserts $n$ items 
into a table $T$ with 
$2(1+\epsilon)n$ buckets (for a fixed constant $\epsilon>0$), via hash functions $h_1$ and $h_2$ whose values are always distinct.
In a standard cuckoo hash table,
each bucket can
store at most one item and each item is stored in one-out-of-two
locations. To insert an item $x$, we place it in
$T[h_1(x)]$ if that bucket is empty.  Otherwise, we \emph{evict} the
item $y$ in $T[h_1(x)]$, replace it with $x$, and attempt to insert
$y$ into its other cell $T[h_1(y)+h_2(y)-h_1(x)]$.  If that location is free, then we are done. If not, we evict the item $z$ in that location, attempt to
insert $z$ into its other cell, and so on.  We can view the hash
functions as defining a random graph with $m$
vertices corresponding to the buckets in $T$, with each of the $n$ items $x$
yielding an edge $(h_1(x),h_2(x))$.  The insertion procedure
successfully places all $n$ items if and only if each connected
component in this \emph{cuckoo graph} has at most one cycle.
With high probability all connected
components have at most cycle and no 
component is of size larger than $O(\log n)$~\cite{Pagh:2014}.  
In our 2-3 cuckoo table insertion algorithm
described above, suppose 
we allow the eviction process to occur up to $L=D \log n$ times,
for some constant $D$, before declaring a failure.
As above with our 2-3 cuckoo case,
there is a small probability that an insertion cannot
be done or takes longer time than expected.  In such
cases, rather than immediately failing, unplaced items can be placed in a 
small stash.
The 2-3 cuckoo hashing framework we use also uses such a stash, as 
mentioned above, of course.

As in the analysis framework of Eppstein {\it et al.}~\cite{Eppstein2017}, 
our analysis of 2-3 cuckoo hashing uses a hypergraph model.
In the iterative insertion algorithm for adding items to a table 
with 2-3 cuckoo hashing, we insert $n$ items into a table $T$ via three hash functions $h_1$, $h_2$, and $h_3$ with distinct values. Each item is
stored at two of the three locations given by its hash values, except that an item that cannot be stored successfully may be placed into a stash.
In this setting, we require the table to have $6(1+\epsilon)n$ total bucket spaces,
so that the final ``load'' of the hash table is less than $1/6$.
It is known that a load strictly less than 
$1/6$ allows all items to be placed with high probability, but
a load strictly greater than $1/6$ will fail to place all items
with high probability.  These results are discussed and proven 
by Amossen and Pagh \cite{Amossen11} and Loh and Pagh \cite{loh2014thresholds}
(see also the related combinatorial results in \cite{karonski2002phase}).  

The generalization of the cuckoo graph for standard cuckoo hashing to
the two-out-of-three paradigm leads to
a {\em cuckoo hypergraph} model~\cite{Eppstein2017}, 
where each bucket is represented by a vertex and each item
is represented by a hyperedge attached to the three vertices (buckets)
that are the chosen locations for that item.  
Eppstein {\it et al.}~\cite{Eppstein2017} use this model to show the
following:

\begin{theorem}
\label{thm:failure} For any constant integer $s \ge 1$, for a
sufficiently large constant $C$, the size $S$ of the stash in
a 2-3 cuckoo hash table after
all items have been inserted satisfies $\Pr(S \ge s) = \tilde{O}(n^{-s})$.
\end{theorem}

The $\tilde{O}$ notation ignores polylogarithmic factors.
Their proof of this theorem uses 
a 3-uniform hypergraph $(V,T)$, which consists of a set $V$ of vertices and a
set or multiset $T$ of triples of vertices, the hyperedges of the
cuckoo hypergraph. It can be represented by a bipartite \emph{incidence
  graph} $(V,H,I)$, a graph that has $V\cup H$ as its vertices and the
set of vertex-hyperedge incidences $I=\{(v,t)\mid v\in V, t\in T, v\in
t\}$ as its edges. Even if the hypergraph has repeated
pairs or triples among its vertices, the incidence graph is a simple
graph. Following Berge~\cite{Berge}, we call a hypergraph \emph{acyclic} if its incidence graph is acyclic as an
undirected graph. The \emph{connected components} of a hypergraph are
the subsets of vertices and hyperedges corresponding to connected
components of the incidence graph. A hypergraph is \emph{connected} if
it has exactly one connected component. We define acyclicity of
components in the same way as acyclicity of the whole
hypergraph. We say that a connected component of the hypergraph is
\emph{unicyclic} if the corresponding connected component of the
incidence graph is unicyclic: that is, that it has exactly one
undirected cycle. 
Such connected components allow for successful 2-3 cuckoo insertions in their 
corresponding locations, but the running times for such constructions
depend on the sizes of such components.
Let $C_v$ be the component containing $v$ in the randomly chosen hypergraph,
and let $E_v$ represent the set of edges in $C_v$.  
Eppstein {\it et al.}~\cite{Eppstein2017} show the following.

\begin{lemma}
\label{lem:expdec}
There exists a constant $\beta \in (0,1)$ such that for any fixed vertex $v$ and integer $k > 0$, 
$$\Pr(|E_v| \geq k) \leq \beta^k.$$
\end{lemma}

From this result, we can derive the following bound on the total time
needed for
our construction algorithm given above, following an approach
used by Goodrich and Mitzenmacher to analyze a parallel 
algorithm for standard cuckoo hashing~\cite{Goodrich2011b}.

\begin{theorem}
	The running time of the iterative insertion algorithm
	for constructing a 2-3 cuckoo hash table
	is $O(n)$ with high probability.
\end{theorem}
\begin{proof}
The total time for performing the iterative insertion algorithm 
is proportional to $\sum_v |E_v|$, which 
	in expectation is as follows, by Lemma~\ref{lem:expdec}:
\begin{eqnarray*}
	{\bf E}\left[\sum_v |E_v|\right] &=&
	\sum_v {\bf E}[|E_v|] \\	
	&\le& 2n \sum_{k\ge 0} \Pr(|E_v|\ge k) \\
	&\le& 2n \sum_{k\ge 0} \beta^k \\
	&=& O(n).
\end{eqnarray*}
For the high probability bound, we follow
a similar argument used by Goodrich and Mitzenmacher~\cite{Goodrich2011b}
for standard cuckoo hashing,
which uses a variant of 
Asuma's inequality
If all component sizes were bounded by say $O(\log^2 n)$, 
then a change in any single edge in the cuckoo hypergraph could affect
$\sum_v |E_v|$ by
only $O(\log^6 n)$, and we could directly apply Azuma's
inequality to the
Doob martingale obtained by exposing the edges of the cuckoo hypergraph 
one at a time. Unfortunately, all
component sizes are $O(\log^2 n)$ only with high probability. 
However, standard results yield that one
can simply add in the probability of a ``bad event'' to 
a suitable tail bound, in this case the bad event
being that some component size is larger than 
$c_1 \log^2 n$ for some suitable constant $c_1$. 
Specifically, we
directly utilize Theorem 3.7 from McDairmid~\cite{mcdiarmid1998},
which allows us to conclude that if the probability of a bad event is
a superpolynomially small $\delta$, then
\[
	\Pr\left(\sum_v |E_v| \ge \sum {\bf E}[|E_v|] + \lambda\right)
	\le e^{-(2\lambda^2)/(nc_2\log^6 n)} + \delta,
\]
where $c_2$ is a suitable constant.
\end{proof}

\end{appendix}
\fi

\end{document}